\documentclass[11pt]{article}
\usepackage[margin=1in]{geometry}
\usepackage{amsmath,amssymb,amsthm,mathtools}
\usepackage[T1]{fontenc}
\usepackage{lmodern}
\usepackage[hidelinks]{hyperref}

\usepackage{tikz}
\usetikzlibrary{positioning,arrows.meta,calc}
\newtheorem{theorem}{Theorem}
\newtheorem{lemma}[theorem]{Lemma}
\newtheorem{proposition}[theorem]{Proposition}
\newcommand{\MMS}{\operatorname{MMS}}
\newcommand{\ind}{\mathbf 1}
\newcommand{\cC}{\mathcal C}
\newcommand{\alloc}{\mathcal{A}}
\newcommand{\items}{M}
\newcommand{\reals}{\mathbb{R}}
\usepackage[]{color-edits}
\addauthor{TE}{red}

\usepackage{listings}

\usepackage[
  backend=biber,
  style=alphabetic,
  natbib=true,
  maxcitenames=99,
  maxbibnames=99,
  maxalphanames=3,
  minalphanames=3,
  giveninits=true
]{biblatex}

\usepackage{orcidlink}

\date{}
\title{Improved Impossibility Bounds for Maximin Share Allocations}
\author{%
  \normalsize\mdseries
  Tomer Ezra\,\orcidlink{0000-0003-0626-4851}\thanks{%
    Tel Aviv University, Tel Aviv, Israel.
    Email: \texttt{tomerezra@tauex.tau.ac.il}.%
  }
  \and
  \normalsize\mdseries
  Tamar Garbuz\,\orcidlink{0009-0009-3922-4160}\thanks{Independent researcher.
}
}
\begin{document}
\maketitle
\begin{abstract}
The maximin share (MMS) is a central fairness benchmark for allocating indivisible items, but it need not be simultaneously attainable even under additive preferences. While extensive work has developed approximation guarantees, quantitative impossibility bounds have received comparatively little attention. We establish improved asymptotic and constant impossibility bounds for both goods and chores.

For every sufficiently large number $n$ of agents, we construct additive goods instances in which every allocation gives some agent at most a $1-\Omega((\log n)^{-2})$ fraction of her MMS. This strengthens the $1/n^4$ shortfall of Feige, Sapir, and Tauber (2021) to an inverse-polylogarithmic shortfall. For chores, we construct instances in which every allocation gives some agent cost at least a $1+\Omega((\log n)^{-2})$ factor of her MMS. Consequently, for every fixed $\varepsilon>0$, guarantees of $1-O(n^{-\varepsilon})$ for goods and $1+O(n^{-\varepsilon})$ for chores are impossible. Both asymptotic constructions use $n+O(\log n)$ items. Together with concurrent existence results of Feige (2026), this determines the threshold $n+\Theta(\log n)$ for guaranteed MMS existence under additive preferences, for both goods and chores.

We also give four-agent, eleven-item instances that improve the universal impossibility bounds from $39/40$ to $20/21$ for goods and from $44/43$ to $31/30$ for chores.
\end{abstract}

\newpage

\section{Introduction}

Fair division studies how to allocate resources among agents with
different preferences. For $n$ agents with equal entitlements, a
natural starting point is proportionality, which requires each agent
to receive at least a $1/n$ fraction of her value for the entire set
of goods.
With indivisible items, however, this requirement can be infeasible
even when all agents have identical valuations. A single valuable
item shared among two agents is the simplest example. The challenge
is therefore to identify fairness benchmarks that account for
indivisibility while retaining a meaningful individual guarantee.

The \emph{maximin share} (MMS), introduced by
\citet{budish2011combinatorial}, provides such a benchmark. An agent
partitions the goods into $n$ bundles and then receives a least-valued
bundle according to her own valuation. Her MMS is the largest value
she can secure in this procedure. The benchmark depends only on her
own preferences and the available goods, and is always simultaneously
attainable when all agents have identical valuations. With different
valuations, however, the partitions witnessing the agents' individual
guarantees may be incompatible. Indeed,
\citet{kurokawa2018fair} showed that an MMS allocation need not exist
even for three agents with additive valuations. This leads to a
quantitative question: what fraction of every agent's MMS can always
be guaranteed simultaneously?

A substantial literature has developed increasingly strong answers on
the positive side. Starting from the $2/3$ guarantee of
\citet{kurokawa2018fair}, subsequent work improved the computational
efficiency and simplified the algorithms for obtaining approximate
MMS allocations
\citep{amanatidis2017approximation,barman2020approximation,
garg2018approximating}.
\citet{ghodsi2018fair} established a $3/4$ guarantee, and
\citet{garg2021improved} gave a strongly polynomial-time algorithm for
$3/4$-MMS allocations as well as an existence guarantee of
$3/4+1/(12n)$. Further refinements improved the dependence on $n$ and
simplified the underlying arguments
\citep{akrami2023simplification,akrami2023improving}.
\citet{akrami2024breaking} obtained the first constant improvement
beyond $3/4$, proving a guarantee of $3/4+3/3836$.
More recently, \citet{heidari2026improved} established a $10/13$
guarantee, and \citet{huang2025fptas} improved this to $7/9$, together
with an algorithm computing a $(7/9-\varepsilon)$-MMS allocation in
time polynomial in the input size and $1/\varepsilon$.

The analogous problem for indivisible \emph{chores} replaces values
by nonnegative additive costs. Each agent seeks to minimize her cost;
her MMS is the smallest possible maximum bundle cost in a partition
of the chores into $n$ bundles. Thus, an approximate MMS allocation
should give every agent cost at most a factor $\beta\ge 1$ of her MMS.
\citet{aziz2017algorithms} showed that exact MMS allocations may fail
to exist and provided a factor-$2$ approximation. This factor was
improved to $4/3$ by \citet{barman2020approximation}, to $11/9$ by
\citet{huang2019algorithmic}, and to $13/11$ by
\citet{huang2023reduction}. The latter work also gives an algorithm
for $(13/11+\varepsilon)$-MMS allocations with running time polynomial
in the input size and $1/\varepsilon$. Connections to job scheduling
and bin packing have been instrumental in these developments
\citep{huang2019algorithmic,huang2023reduction}.

MMS fairness has also been studied beyond multiplicative approximation.
Related directions include ordinal relaxations of the MMS benchmark
\citep{hosseini2022ordinalgoods},
guaranteeing MMS to a subset of the agents
\citep{hosseini2021guaranteeing}, and share-based fairness under unequal
entitlements \citep{babaioff2024fair}. Our focus is on the simultaneous
attainability of multiplicative MMS guarantees in the basic additive
setting, with equal entitlements and without computational
restrictions.

In contrast to the extensive progress on positive guarantees, upper
bounds on the attainable fraction of MMS for goods, and the
corresponding impossibility bounds for chores, have received
comparatively little attention. A notable exception is the work of
\citet{feige2021tight}, who constructed an instance with three agents
and nine goods in which every allocation gives some agent value at
most $39/40$ of her MMS. For chores, they constructed an instance
with three agents and nine chores in which every allocation gives
some agent cost at least $44/43$ of her MMS. These examples demonstrate
that exact MMS fairness can fail by a nonzero constant, even in small
instances. They do not, however, resolve how much of the gap between
the positive guarantees and exact MMS fairness is unavoidable.

The dependence on the number of agents raises a distinct question.
A counterexample with a fixed number of agents bounds the approximation
factor that can hold uniformly over all instances, but does not by
itself quantify the unavoidable loss for every sufficiently large $n$.
For each $n\ge 4$, \citet{feige2021tight} constructed an additive goods
instance in which every allocation gives some agent at most a
$1-1/n^4$ fraction of her MMS.

To frame the asymptotic question, \citet{feige2021tight} consider the
worst-case relative MMS shortfall $\delta_n$, so that $1-\delta_n$
is the largest fraction of MMS that can be guaranteed simultaneously
to all agents in every $n$-agent additive goods instance. They
explicitly ask:
\begin{quote}
\emph{Does $\delta_n$ tend to zero as the number of agents grows?}
\end{quote}
Equivalently, can every agent be guaranteed a $1-o(1)$ fraction of
her MMS as $n$ grows, uniformly over all additive goods instances?
Their construction establishes $\delta_n\ge 1/n^4$, but does not rule out $\delta_n \rightarrow 0$. Beyond whether the shortfall vanishes,
this raises the quantitative question of how quickly it could vanish,
and hence how severe the failure of MMS fairness can be for large
populations.

In this paper, we strengthen the impossibility bounds for both goods
and chores. We make quantitative progress on the question of
\citet{feige2021tight} by strengthening their inverse-polynomial lower
bound on the worst-case MMS shortfall to an inverse-polylogarithmic
one. We also obtain an asymptotic impossibility result for
chores and improve the constant impossibility bounds in both settings.

\subsection{Model}

We consider a set of $n$ agents, indexed by $[n]=\{1,\ldots,n\}$,
and a set $\items$ of $m$ indivisible items. We study the allocation
of goods and chores separately, beginning with the goods setting.
Each agent $i\in[n]$ has a nonnegative additive valuation
$v_i:2^{\items}\rightarrow\reals_{\geq 0}$. Writing
$v_i(g)$ for $v_i(\{g\})$, additivity means that
\[
    v_i(S)=\sum_{g\in S}v_i(g)
    \qquad \text{for every } S\subseteq\items.
\]
An instance is specified by $I=(\items,(v_i)_{i\in[n]})$.

An allocation is an ordered partition
$\mathcal B=(\mathcal B_1,\ldots,\mathcal B_n)$ of $\items$
into $n$ pairwise disjoint, possibly empty bundles, where
$\mathcal B_i$ is assigned to agent $i$. In particular, every
item must be allocated. We denote the set of all such allocations
by $\alloc_n(\items)$.

\paragraph{Maximin share.}
The \emph{maximin share} (MMS) of agent $i$ is
\[
    \MMS_i
    =
    \max_{(\mathcal B_1,\ldots,\mathcal B_n)\in\alloc_n(\items)}
    \min_{t\in[n]}v_i(\mathcal B_t).
\]
Thus, $\MMS_i$ is the largest value that agent $i$ can guarantee
by partitioning the items into $n$ bundles and receiving a
least-valued bundle according to her valuation.

\paragraph{Proportional share.}
The \emph{proportional share} of agent $i$ is
\[
    \mathrm{PROP}_i=\frac{v_i(\items)}{n}.
\]
An allocation $\mathcal B$ is \emph{proportional} if
$v_i(\mathcal B_i)\geq\mathrm{PROP}_i$ for every agent $i\in[n]$.
Since the least-valued bundle in any partition is worth at most
the average bundle value, we have
\[
    \MMS_i\leq\mathrm{PROP}_i
    \qquad \text{for every } i\in[n].
\]

\paragraph{MMS approximation.}
For $\alpha\in[0,1]$, an allocation $\mathcal B$ is
\emph{$\alpha$-MMS} if
\[
    v_i(\mathcal B_i)\geq\alpha\MMS_i
    \qquad \text{for every } i\in[n].
\]
For each number of agents $n$, we define the worst-case MMS
approximation ratio by
\[
    \alpha_n^{\mathrm{MMS}}
    =
    \inf_I
    \max_{(\mathcal B_1,\ldots,\mathcal B_n)\in\alloc_n(\items)}
    \min_{i\in[n]}
    \frac{v_i(\mathcal B_i)}{\MMS_i},
\]
where the infimum ranges over all additive goods instances $I$
with $n$ agents, any finite number of items, and $\MMS_i>0$ for
every agent $i\in[n]$. Equivalently, $\alpha_n^{\mathrm{MMS}}$
is the largest $\alpha$ such that every such instance admits an
$\alpha$-MMS allocation. We further define
\[
    \alpha^{\mathrm{MMS}}
    =
    \inf_{n\geq 1}\alpha_n^{\mathrm{MMS}},
\]
the worst-case MMS approximation ratio over all numbers of agents.

\paragraph{Chores.}
In the chores setting, each agent $i\in[n]$ has a nonnegative
additive cost function
$c_i:2^{\items}\rightarrow\reals_{\geq 0}$, and an instance is
specified by $I=(\items,(c_i)_{i\in[n]})$. As with valuations,
we write $c_i(g)$ for $c_i(\{g\})$, and additivity means that
\[
    c_i(S)=\sum_{g\in S}c_i(g)
    \qquad \text{for every } S\subseteq\items.
\]
Allocations are defined as above; in particular, every chore
must be allocated.

The \emph{maximin share} of agent $i$ for chores is
\[
    \MMS_i^{\mathrm{ch}}
    =
    \min_{(\mathcal B_1,\ldots,\mathcal B_n)\in\alloc_n(\items)}
    \max_{t\in[n]}c_i(\mathcal B_t).
\]
Thus, $\MMS_i^{\mathrm{ch}}$ is the smallest cost that agent $i$
can guarantee by partitioning the chores into $n$ bundles and
receiving a highest-cost bundle according to her cost function.

The \emph{proportional share} of agent $i$ for chores is
\[
    \mathrm{PROP}_i^{\mathrm{ch}}=\frac{c_i(\items)}{n}.
\]
A chores allocation $\mathcal B$ is \emph{proportional} if
$c_i(\mathcal B_i)\leq\mathrm{PROP}_i^{\mathrm{ch}}$ for every
agent $i\in[n]$. Since the highest-cost bundle in any partition
has cost at least the average bundle cost, we have
\[
    \MMS_i^{\mathrm{ch}}\geq\mathrm{PROP}_i^{\mathrm{ch}}
    \qquad \text{for every } i\in[n].
\]

\paragraph{MMS approximation for chores.}
For $\beta\geq 1$, a chores allocation $\mathcal B$ is
\emph{$\beta$-MMS} if
\[
    c_i(\mathcal B_i)\leq\beta\MMS_i^{\mathrm{ch}}
    \qquad \text{for every } i\in[n].
\]
For each number of agents $n$, we define the worst-case MMS
approximation factor for chores by
\[
    \beta_n^{\mathrm{MMS}}
    =
    \sup_I
    \min_{(\mathcal B_1,\ldots,\mathcal B_n)\in\alloc_n(\items)}
    \max_{i\in[n]}
    \frac{c_i(\mathcal B_i)}{\MMS_i^{\mathrm{ch}}},
\]
where the supremum ranges over all additive chores instances $I$
with $n$ agents, any finite number of chores, and
$\MMS_i^{\mathrm{ch}}>0$ for every agent $i\in[n]$.
Equivalently, $\beta_n^{\mathrm{MMS}}$ is the smallest $\beta$
such that every such instance admits a $\beta$-MMS allocation.
We further define
\[
    \beta^{\mathrm{MMS}}
    =
    \sup_{n\geq 1}\beta_n^{\mathrm{MMS}},
\]
the worst-case MMS approximation factor for chores over all
numbers of agents. In either setting, an \emph{MMS allocation}
is an allocation with approximation factor $1$.

Throughout the paper, all logarithms are to base $2$.

\subsection{Our Contribution}

We establish two improved impossibility results for additive goods and corresponding results for additive chores.
The first strengthens the unavoidable shortfall from inverse polynomial gap to inverse polylogarithmic in the number of agents. Our second result improves the constant upper bound on the universal
MMS approximation ratio.

\paragraph{An inverse-polylogarithmic shortfall.}
Our first result shows that, for every sufficiently large number $n$
of agents, there exists an instance in which every allocation gives
some agent at most a $1-\Omega((\log n)^{-2})$ fraction of her MMS.

\begin{theorem}\label{thm:main}
For every sufficiently large integer $n$,
\[
\alpha_n^{\mathrm{MMS}}
\le 1-\Omega\bigl((\log n)^{-2}\bigr).
\]
The bound is witnessed by an additive goods instance with
$n+O(\log n)$ items.
\end{theorem}

This result strengthens the inverse-polynomial shortfall established
by \citet{feige2021tight} to an inverse-polylogarithmic shortfall. In particular,
it rules out a $1-O(n^{-\varepsilon})$ MMS guarantee for every fixed
$\varepsilon>0$. The construction satisfies
$\MMS_i=\mathrm{PROP}_i$ for every agent $i$. Thus, the impossibility
persists even in instances where every agent can individually
partition the items into $n$ bundles of equal value. 

Our construction assigns agents distinct types using a
constant-weight error-correcting code of length $\Theta(\log n)$.
Each valuation combines common auxiliary item weights with small
type-dependent corrections. The auxiliary weights force any
allocation that gives every agent sufficiently close to her MMS
to have a restricted bundle structure. The corrections then imply
that the codewords associated with two distinct recipients must
be too close, contradicting the minimum distance of the code.
Quantitatively, every agent has MMS of $\Theta((\log n)^3)$, whereas
every allocation gives some agent an additive shortfall of
$\Omega(\log n)$, yielding the claimed relative loss.

The number of goods in our construction is also asymptotically
optimal in terms of the excess over $n$.
\citet{feige2026mms} proves that, for some constant $c>0$ and
every sufficiently large $n$, every instance with $n$ agents
and at most $n+c\log n$ goods admits an MMS allocation, even
under general monotone valuations. 
Thus, together with the existence guarantee of
\citet{feige2026mms}, our result settles the asymptotic
item threshold for additive valuations, showing that it
is $n+\Theta(\log n)$, as for general monotone valuations.

\paragraph{An improved constant upper bound.}
Our second result improves the upper bound on the universal MMS
approximation ratio from $\frac{39}{40}$, established by
\citet{feige2021tight}, to $\frac{20}{21}$.

\begin{theorem}\label{thm:main2}
It holds that
\[
    \alpha^{\mathrm{MMS}}
    \leq \frac{20}{21}.
\]
\end{theorem}

Consequently, no approximation factor strictly greater than
$\frac{20}{21}$ can be guaranteed for all additive instances.
Whereas Theorem~\ref{thm:main} quantifies the unavoidable shortfall
for every sufficiently large number of agents,
Theorem~\ref{thm:main2} strengthens the constant upper bound on
the guarantee that holds uniformly over all numbers of agents.
The bound is witnessed by an instance with four agents, two
valuation types, and eleven goods; this number of goods is
minimal for four agents, since \citet{hummel2023lower}
guarantees MMS existence with at most ten goods.

\paragraph{Extensions to chores.}
Our constructions also extend to the allocation of indivisible
chores. First, we adapt the proof of Theorem~\ref{thm:main} to
show that an inverse-polylogarithmic excess over the MMS is
unavoidable: for every sufficiently large number $n$ of agents,
there exists an instance in which every allocation gives some
agent cost at least a $1+\Omega((\log n)^{-2})$ factor of her MMS.

\begin{theorem}\label{thm:chores-hardness}
For every sufficiently large integer $n$,
\[
    \beta_n^{\mathrm{MMS}}
    \geq 1+\Omega\bigl((\log n)^{-2}\bigr).
\]
The bound is witnessed by an additive chores instance with
$n+O(\log n)$ items.
\end{theorem}
In particular, this result rules out a
$1+O(n^{-\varepsilon})$ MMS approximation for every fixed
$\varepsilon>0$.
As in the goods case, our construction and the existence
result of \citet{feige2026mms} yield a tight threshold of
$n+\Theta(\log n)$ chores under additive preferences,
matching the asymptotic order for general monotone preferences.

We also obtain a $31/30$ lower bound on the universal MMS
approximation factor for chores.
\begin{theorem}\label{thm:constant-chores}
It holds that
\[
\beta^{\mathrm{MMS}} \ge \frac{31}{30}.
\]
\end{theorem}
This bound is witnessed by an instance with four agents and eleven
chores. Thus, no approximation factor strictly smaller than $31/30$
can be guaranteed for all additive chores instances.

\subsection{Further Related Work}
\label{sec:related-work}

\paragraph{Restricted instances.}
Early work studied conditions under which exact MMS allocations exist
\citep{bouveret2016characterizing,kurokawa2016can}.
Stronger approximation guarantees are available for three agents
\citep{feige2022improved} and for goods instances with only two distinct
valuation types \citep{shahkar2025improved}. For an ordinal relaxation,
\citet{schwerdtfeger20261} established that a $1$-out-of-$5$ MMS
allocation always exists for four additive agents.
The number of items also affects exact MMS existence.
\citet{feige2021tight} proved that an MMS allocation exists whenever
there are at most $n+5$ goods. More generally,
\citet{hummel2023lower} showed that, for every fixed positive integer
$c$, an MMS allocation exists in every instance with at most $n+c$
items once $n$ is sufficiently large, both for goods and for chores ($n \geq 0.6597^c \cdot c!$ for goods and $n \geq 0.7838^c \cdot c!$ for chores).
Concurrently to our work, \citet{feige2026mms} proved that every instance
with at most $n+c_0\log n$ items admits an MMS allocation,
for some constant $c_0>0$ and every sufficiently large $n$,
both for goods and for chores, even under general monotone
preferences. Our counterexamples show that this logarithmic
excess is asymptotically optimal already under additivity.

\paragraph{Beyond additive valuations.}
MMS guarantees have also been studied for submodular valuations
\citep{ghodsi2018fair,uziahu2023fair}, fractionally subadditive valuations
\citep{seddighin2024improved,akrami2023randomized}, and subadditive
valuations \citep{seddighin2024improved,feige2025concentration,
seddighin2025beating}. Our results concern additive valuations,
showing that substantial barriers to simultaneous MMS approximation
already arise without nonadditive preferences.

\subsection{Preliminaries}

Our proofs use binary error-correcting codes. A \emph{binary code}
of length $k$ is a set $\cC\subseteq\{0,1\}^k$, whose elements are
called \emph{codewords}. We index the coordinates of each codeword
by $0,\ldots,k-1$. The \emph{Hamming distance} between two words
$a,b\in\{0,1\}^k$ is the number of coordinates in which they differ:
\[
    d_H(a,b)
    =
    \bigl|\{h\in\{0,\ldots,k-1\}:a_h\neq b_h\}\bigr|.
\]
The \emph{minimum Hamming distance} of a code is the minimum of
$d_H(a,b)$ over all distinct codewords $a,b\in\cC$.

The \emph{Hamming weight} of a word is the number of its coordinates
equal to $1$. A code is \emph{constant-weight} if all its codewords
have the same Hamming weight. We use codes whose codewords are
\emph{balanced}, meaning that exactly half of their coordinates
are equal to $1$. For even $k$, let
\[
    L_k
    =
    \left\{
        a\in\{0,1\}^k:
        \sum_{h=0}^{k-1}a_h=\frac{k}{2}
    \right\}
\]
denote the set of balanced words of length $k$.

The following lemma guarantees the existence of exponentially many
balanced codewords with pairwise Hamming distance linear in their
length. We defer its proof to the appendix.

\begin{lemma}\label{lem:code}
For every integer $k\geq128$ divisible by $8$, there exists a code
$\cC\subseteq L_k$ with $|\cC|\geq 2^{k/8}$ and minimum Hamming
distance at least $k/4$.
\end{lemma}

\section{Asymptotic Hardness for Goods} \label{sec:goods}
\subsection{Description of the Instance}
Since our hardness result is asymptotic, it is sufficient to consider a large enough $n$.
Let $n\geq 2^{16}$ and let $k = 8\lceil \log n\rceil \geq 128$.  We consider instances with $m= n+5k$ items.
Let $\cC$ be an error-correcting code satisfying the requirements of Lemma~\ref{lem:code}.
Associate with a codeword $a\in\cC$ the label set
\[
 S(a)=\{2h+a_h:0\le h<k\}\subseteq\{0,\ldots,2k-1\}.
\]
These sets have size $k$, satisfy
\begin{equation}\label{eq:sum}
 \sum_{j\in S(a)}j=k(k-1)+k/2,
\end{equation}
and, for distinct words $a,b$, satisfy the one-sided distance bound
\begin{equation}\label{eq:distance}
 |S(a)\setminus S(b)|=d_H(a,b)\ge k/4.
\end{equation}

\paragraph{Goods and valuations.}
Choose $n$ distinct codewords from $\cC$, one per agent (we can do so since $n\leq |\cC|$). We associate an agent with her corresponding label set. 
Set
\[
 x_j=k(6k+j),\qquad H=17k^2,\qquad T=4kH=68k^3,
 \qquad U=7k^3-k^2/2.
\]
Equation~\eqref{eq:sum} gives $\sum_{j\in S}x_j=U$ for every type $S$.
We define auxiliary weights as follows:
\[
\begin{array}{c|c|c}
\text{good}&\text{indices or number}&w(g)\\ \hline
 Z_j&0\le j<2k&T-H\\
 X_j&0\le j<2k&x_j\\
 Y_j&0\le j<2k&H-x_j\\
 G_h&0\le h<k&H\\
 P&1&T-U\\
 Q&1&T-kH+U\\
 D_t&1\le t\le n-2k-2&T
\end{array}
\]
This is well defined since $n\geq 2k+2$.
All auxiliary weights are strictly positive and are multiples of $k$. The total auxiliary weight is $nT$.
For an agent of type $S$, we define correction functions by
\[
 \rho_S(X_j)=-\ind[j\notin S],\qquad
 \rho_S(Y_j)=\ind[j\in S],\qquad
 \rho_S(Z_j)=1,
\]
\[
 \rho_S(G_h)=-1,\qquad \rho_S(P)=0,\qquad
 \rho_S(Q)=-k,\qquad \rho_S(D_t)=0.
\]
Note that each correction function, when aggregated over all items, sums to  $0$. We say that an $X_j$ good or a $Y_j$ good is labeled with respect to a label set $S$ if $j\in S$. 

Finally, the valuation of an agent with label set $S$ is:
\begin{equation}\label{eq:values}
 v_S(g)=3w(g)+\rho_S(g).
\end{equation}
Every item value is a strictly positive integer.
The proportional share of every agent is $\mu=3T=204k^3$.

For any set $B$ of goods, we write
\[
 w(B)=\sum_{g\in B}w(g),\qquad
 \rho_S(B)=\sum_{g\in B}\rho_S(g).
\]
Thus $v_S(B)=3w(B)+\rho_S(B)$. 

\paragraph{Intuition of the construction.}
The common auxiliary weights determine the shape of any
allocation in which everyone is close to her MMS. Each
$D$-item must be alone. Each $Z_j$-item must be accompanied
either by one $G$-item or by the matching pair $\{X_j,Y_j\}$.
The remaining items must go into the two bundles containing
$P$ and $Q$.

Now consider the agent receiving $Q$. For her, the $Y$-items
matching her type provide positive corrections, which are
needed to offset the negative correction of $Q$. She must
therefore receive almost all of these $Y$-items. Whenever
such a $Y_j$ goes to $Q$, it is unavailable to accompany
$Z_j$, so $Z_j$ must use a $G$-item instead, leaving $X_j$
to be placed elsewhere. Since the selected $Y$-items already
fill almost all the available space alongside $Q$, almost
all of their matching $X$-items must go to the bundle
containing $P$.

The agent receiving $P$, however, cannot afford many
$X$-items that do not match her own type. Thus, almost
every label belonging to the type of the agent receiving
$Q$ must also belong to the type of the agent receiving
$P$. Since all types have the same number of labels, these
two types must almost agree. This contradicts the code's
distance property, which ensures that distinct agents'
types differ on a constant fraction of their labels.

\begin{lemma}\label{lem:mms}
Every agent has MMS exactly $\mu$. 
\end{lemma}
\begin{proof}
Fix a type $S$. For each $j\in S$, pair $Z_j$ with a different $G_h$.
For each $j\notin S$, use $\{Z_j,X_j,Y_j\}$. Put every $D_t$ in a bundle by itself. The two remaining 
bundles are
\[
 \{P\}\cup\{X_j:j\in S\},\qquad
 \{Q\}\cup\{Y_j:j\in S\}.
\]
Each bundle has auxiliary weight $T$ and correction zero: the corrections
are respectively $1-1$, $1-1+0$, $0$, $0$, and $-k+k$.
Overall, these are $n$ bundles partitioning
all goods. Therefore, the total auxiliary weight is $nT$, the total value is $n\mu$,
and this partition witnesses MMS at least $\mu$. The MMS is no more than $\mu$ since this is the proportional share, and the MMS cannot exceed the proportional share for additive valuations.
\end{proof}

\subsection{Necessary Conditions for an Approximate MMS Allocation}
\begin{lemma}\label{lem:pairs}
A subset of low goods (goods of type $X$, $Y$, or $G$) has auxiliary weight $H$ if and only
if it is a singleton $\{G_h\}$ or a matched pair $\{X_j,Y_j\}$ (with the same index).
\end{lemma}
\begin{proof}
The ranges are
\[
 6k^2\le w(X_j)\le8k^2-k,
 \qquad 9k^2+k\le w(Y_j)\le11k^2.
\]
A $G$-good already supplies $H$. Two $Y$-goods exceed $H$.
With one $Y_j$, the residual auxiliary weight is $x_j<8k^2$, which can only be
supplied by a single $X$-good; distinctness of the $x_j$ forces it to
be $X_j$. Finally, one or two $X$-goods have auxiliary weight below $H=17k^2$,
whereas three or more have auxiliary weight at least $18k^2>H$.
\end{proof}

\begin{proposition}\label{prop:robust}
Every allocation gives some agent value strictly below $\mu-k/24$.
\end{proposition}
\begin{proof}
Put $\ell=k/24$ and suppose, for contradiction, that every agent receives
at least $\mu-\ell$. 

\paragraph{Exact auxiliary weights.}
Call $Z_j,P,Q,D_t$ high goods (and the remaining low goods); there are exactly $n$ of them.
All low goods together have a total auxiliary weight of $3kH$.
For any agent, the only positive low-good corrections are on her $k$
labeled $Y$-goods. Thus a low-only bundle has value at most
\[
 9kH+k=153k^3+k<\mu-\ell.
\]
Every bundle must contain a high good, so each contains exactly one.
Its correction is at most $k+1$. All auxiliary weights are multiples of
$k$, and hence a bundle of auxiliary weight below $T$ has value at most
\[
 3(T-k)+k+1=\mu-2k+1<\mu-\ell.
\]
Every bundle therefore has auxiliary weight at least $T$. Their total auxiliary weight is
$nT$, so each has auxiliary weight exactly $T$. In particular, the $D_t$ items
are singletons, and every allocated bundle has correction at least $-\ell$.

\paragraph{The bundles containing $P$ and $Q$.}
Let $\mathcal B_P$ be the allocated bundle containing the good $P$, and
let $\mathcal B_Q$ be the allocated bundle containing the good $Q$.
Write $S_P$ and $S_Q$ for the types of their respective recipients.
These bundles have different recipients: if one bundle contained both
$P$ and $Q$, it would contain two high goods, contrary to what we just
proved. In particular, $S_P\ne S_Q$, since all agents have distinct types.
We have
\[
 w(\mathcal B_P)=w(\mathcal B_Q)=T,
 \qquad
 \rho_{S_P}(\mathcal B_P)\ge-\ell,
 \qquad
 \rho_{S_Q}(\mathcal B_Q)\ge-\ell.
\]

For each $j$, the bundle containing $Z_j$ has total auxiliary weight $T$.
Since $w(Z_j)=T-H$, its low goods must have total auxiliary weight exactly $H$.
Lemma~\ref{lem:pairs} therefore gives exactly two possibilities: its low
goods are a singleton $\{G_h\}$ or a matched pair $\{X_{j'},Y_{j'}\}$.
The pair's label $j'$ need not equal $j$.

Let
\[
 g=\bigl|\{h:G_h\in\mathcal B_P\cup\mathcal B_Q\}\bigr|
\]
be the number of $G$-goods in the two distinguished bundles. There are
$k$ $G$-goods in total. Every one of the other $k-g$ $G$-goods belongs to
a bundle containing a $Z$-good, because $D_t$ goods are singletons and
all goods are allocated. Each such bundle contains exactly one $G$-good.
Hence exactly $k-g$ of the $2k$ bundles containing $Z$-goods use a
$G$-good. Each of the remaining
\[
 2k-(k-g)=k+g
\]
bundles uses one matched pair. 

Define the set of labels of $X$ and $Y$ items that are in $\mathcal B_P\cup\mathcal B_Q$  by
\[
 F=\{j\in\{0,\ldots,2k-1\}:
                  X_j,Y_j\in\mathcal B_P\cup\mathcal B_Q\}.
\]
Note that either both $X_j$ and $Y_j$ are in $\mathcal B_P\cup\mathcal B_Q$ or both are not in $\mathcal B_P\cup\mathcal B_Q$. It holds that
\[
 |F|=2k-(k+g)=k-g,
\]
and the low goods in $\mathcal B_P\cup\mathcal B_Q$ are exactly
\[
 \{X_j,Y_j:j\in F\}
 \quad\text{together with  }g\text{  }G\text{-goods}.
\]

\paragraph{Few labeled $Y$-goods can be missing from $\mathcal B_Q$.}
The recipient of $\mathcal B_Q$ has type $S_Q$ and has exactly $k$
labeled $Y$-goods, namely $\{Y_j:j\in S_Q\}$. Define
\[
 I_Q=\{j\in S_Q:Y_j\in\mathcal B_Q\},
 \qquad d=k-|I_Q|.
\]
Thus $d$ counts the labeled $Y$-goods missing from $\mathcal B_Q$.
Let $\xi_Q$ be the number of goods $X_j\in\mathcal B_Q$ with
$j\notin S_Q$, and let $g_Q$ be the number of $G$-goods in
$\mathcal B_Q$. The good $Q$ contributes correction $-k$; the
$|I_Q|=k-d$ labeled $Y$-goods contribute $+1$ each; the $\xi_Q$
unlabeled $X$-goods and the $g_Q$ $G$-goods contribute $-1$ each.
All other goods in this bundle have correction zero. Consequently,
\[
 \rho_{S_Q}(\mathcal B_Q)
 =-k+(k-d)-\xi_Q-g_Q
 =-d-\xi_Q-g_Q.
\]
Since this correction is at least $-\ell$, we obtain
\[
 d+\xi_Q+g_Q\le \ell,
 \qquad\text{and in particular}\qquad 0\le d\le \ell.
\]

Every label in $I_Q$ belongs to $F$, because its $Y$-good is in
$\mathcal B_Q$. Thus $I_Q\subseteq S_Q\cap F$. This inclusion gives
two separate consequences. First, a label in $S_Q\setminus F$ cannot
belong to $I_Q$, so
\[
 S_Q\setminus F\subseteq S_Q\setminus I_Q,
 \qquad |S_Q\setminus I_Q|=d.
\]
Second, the $k-d$ labeled $Y$-goods in $\mathcal B_Q$ are among the
$|F|=k-g$ $Y$-goods available in $\mathcal B_P\cup\mathcal B_Q$.
Therefore $k-d\le k-g$, which is equivalent to $g\le d$.
Together these two observations prove
\begin{equation}\label{eq:missing}
 |S_Q\setminus F|\le d,\qquad g\le d.
\end{equation}

\paragraph{Few $X$-goods can be in $\mathcal B_Q$.}
The low goods in $\mathcal B_Q$ have total auxiliary weight
\[
 w(\mathcal B_Q\setminus\{Q\})
 =T-w(Q)
 =kH-U.
\]
The full collection of labeled $Y$-goods for this recipient has
exactly the same auxiliary weight:
\[
 \sum_{j\in S_Q}w(Y_j)
 =\sum_{j\in S_Q}(H-x_j)
 =kH-\sum_{j\in S_Q}x_j
 =kH-U.
\]
Let
\[
 \mathcal E_Q
 =\mathcal B_Q\setminus
       \bigl(\{Q\}\cup\{Y_j:j\in I_Q\}\bigr).
\]
This is the set of low goods in $\mathcal B_Q$ other than its labeled
$Y$-goods. Subtracting the auxiliary weight of the labeled $Y$-goods actually
present from the preceding two equal totals gives
\begin{align*}
 w(\mathcal E_Q)
 &=w(\mathcal B_Q\setminus\{Q\})
       -\sum_{j\in I_Q}w(Y_j)\\
 &=\sum_{j\in S_Q}w(Y_j)-\sum_{j\in I_Q}w(Y_j)\\
 &=\sum_{j\in S_Q\setminus I_Q}w(Y_j).
\end{align*}
The final sum contains exactly $d$ terms, each at most $11k^2$.
On the other hand, every good in $\mathcal E_Q$ is a low good and
has auxiliary weight at least $6k^2$. It follows that
\[
 6k^2|\mathcal E_Q|\le w(\mathcal E_Q)\le 11k^2d,
 \qquad
 |\mathcal E_Q|\le\frac{11}{6}d\le2d.
\]
Every $X$-good in $\mathcal B_Q$ belongs to $\mathcal E_Q$.
We have therefore proved
\begin{equation}\label{eq:qA}
 \bigl|\{j:X_j\in\mathcal B_Q\}\bigr|\le2d.
\end{equation}

\paragraph{Most labels of $S_Q$ have their $X$-goods in $\mathcal B_P$.}
Define
\[
 R=\{j\in S_Q:X_j\in\mathcal B_P\}.
\]
A label $j\in S_Q$ fails to belong to $R$ in exactly one of two ways.
Either $j\notin F$, in which case its $X$-good belongs to a bundle
containing a $Z$-good; or $j\in F$ and its $X$-good is in
$\mathcal B_Q$. There is no third possibility, since for every
$j\in F$ the good $X_j$ is allocated to one of the two distinguished
bundles. Thus
\[
 S_Q\setminus R
 =(S_Q\setminus F)
   \,\dot\cup\,
   \{j\in S_Q\cap F:X_j\in\mathcal B_Q\},
\]
where $\dot\cup$ denotes disjoint union. By Equations~\eqref{eq:missing} and \eqref{eq:qA}, the first set has size at most
$d$ and the second at most $2d$. Hence
\[
 |S_Q\setminus R|\le3d.
\]

\paragraph{Most of these labels must also belong to $S_P$.}
There are exactly $|F|=k-g$ $Y$-goods in
$\mathcal B_P\cup\mathcal B_Q$. We already know that
$\mathcal B_Q$ contains at least $k-d$ of them, namely those with
labels in $I_Q$. Thus
\[
 \bigl|\{j:Y_j\in\mathcal B_P\}\bigr|
 \le(k-g)-(k-d)=d-g\le d.
\]
Let $\xi_P$ count the goods $X_j\in\mathcal B_P$ with $j\notin S_P$,
let $\eta_P$ count the goods $Y_j\in\mathcal B_P$ with $j\in S_P$,
and let $g_P$ count the $G$-goods in $\mathcal B_P$.
The good $P$ has correction zero, and hence
\[
 \rho_{S_P}(\mathcal B_P)=\eta_P-\xi_P-g_P\ge-\ell.
\]
Since $\eta_P$ is at most the total number of $Y$-goods in
$\mathcal B_P$, the preceding bound gives $\eta_P\le d$.
Rearranging the correction inequality now yields
\[
 \xi_P\le\eta_P+\ell-g_P\le d+\ell.
\]
For each label $j\in R\setminus S_P$, its $X$-good lies in
$\mathcal B_P$ by the definition of $R$, and its label is outside
$S_P$. This good is therefore counted by $\xi_P$. So
\[
 |R\setminus S_P|\le\xi_P\le d+\ell.
\]

\paragraph{The final distance contradiction.}
Every label in $S_Q\setminus S_P$ either lies outside $R$, or lies
in $R\setminus S_P$. More explicitly, since $R\subseteq S_Q$,
\[
 S_Q\setminus S_P
 =\bigl((S_Q\setminus R)\setminus S_P\bigr)
       \,\dot\cup\,(R\setminus S_P).
\]
The first set has size at most $|S_Q\setminus R|\le3d$, and the
second has size at most $d+\ell$. Therefore
\begin{align*}
 |S_Q\setminus S_P|
 &\le |S_Q\setminus R|+|R\setminus S_P|\\
 &\le 3d+(d+\ell)\\
 &=4d+\ell\\
 &\le5\ell=\frac{5k}{24}<\frac{k}{4}.
\end{align*}
However, $\mathcal B_P$ and $\mathcal B_Q$ have distinct recipients,
whose distinct types were chosen from the code. The one-sided distance
bound~\eqref{eq:distance} therefore gives
$|S_Q\setminus S_P|\ge k/4$, a contradiction.
\end{proof}

\subsection{Putting It All Together}
\begin{proof}[Proof of Theorem~\ref{thm:main}]
Fix an arbitrary allocation $(\mathcal B_i)_{i=1}^n$.
Every agent has $\MMS_i=\mu=204k^3$ by Lemma~\ref{lem:mms}.
Proposition~\ref{prop:robust} gives an agent $i$ with
$v_i(\mathcal B_i)<\mu-k/24$. Dividing by this agent's MMS yields
\[
 \frac{v_i(\mathcal B_i)}{\MMS_i}
 <1-\frac{k/24}{204k^3}
 =1-\frac{1}{4896k^2}.
\]
Using $k\le9\log(n)$ and $4896\cdot81=396576<400000$, we get
\[
 \frac{1}{4896k^2}
 \ge\frac{1}{396576\log^2(n)}
 >\frac{1}{400000\log^2(n)}.
\]
Substituting this lower bound on the loss proves the claimed
approximation upper bound for the arbitrary allocation.
\end{proof}

\section{Constant Hardness for Goods}
\label{sec:constant-gap}

We construct an instance with four agents and eleven items in which
every agent has MMS equal to $21$, but every allocation gives some
agent value at most $20$. In particular, this establishes
$\alpha_4^{\mathrm{MMS}}\leq 20/21$.

\paragraph{Description of the instance.}

There are two agents with valuation $v_1$ and two agents with
valuation $v_2$. The item set is
\[
    \items=\{a,b,c,d,e,f,g,h,i,j,k\}.
\]
The valuations are additive and specified by the following table.
\[
\begin{array}{c|rrrrrrrrrrr}
    \text{Item} & a & b & c & d & e & f & g & h & i & j & k\\
    \hline
    v_1 & 1 & 17 & 3 & 10 & 1 & 10 & 7 & 7 & 7 & 3 & 18\\
    v_2 & 1 & 14 & 5 & 10 & 1 & 9 & 6 & 7 & 7 & 5 & 19
    \end{array}
\]
Notice that
\[
    v_1(\items)=v_2(\items)=84.
\]

\begin{proof}[Proof of Theorem~\ref{thm:main2}]
For valuation $v_1$, the partition
\[
    \bigl(
        \{a,b,c\},
        \{d,e,f\},
        \{g,h,i\},
        \{j,k\}
    \bigr)
\]
consists of four bundles of value $21$.
For valuation $v_2$, the partition
\[
    \bigl(
        \{d,g,j\},
        \{a,e,k\},
        \{b,h\},
        \{c,f,i\}
    \bigr)
\]
also consists of four bundles of value $21$.
Thus, every agent has MMS at least $21$. Since every agent's
proportional share is $84/4=21$, equality follows.

The exhaustive verification in Appendix~\ref{app:code} shows that every allocation gives some agent value at most $20$.
\end{proof}

\section{Asymptotic Hardness for Chores}
\begin{proof}[Proof of Theorem~\ref{thm:chores-hardness}]
Take the goods construction with valuations
$v_i(g)=3w(g)+\rho_i(g)$, and retain its notation
$\mu=3T$ and $\ell=k/24$. On the same items, define additive costs
\[
 c_i(g)=3w(g)-\rho_i(g).
\]
These costs are strictly positive integers. Every bundle in
agent $i$'s MMS partition (defined for the goods setting in Section~\ref{sec:goods}) has auxiliary weight $T$
and correction zero, so its cost remains $\mu$. Moreover,
$c_i(M)=n\mu$. This partition and the average-cost
lower bound therefore give
\[
 \operatorname{MMS}^{\mathrm{ch}}_i=\mu
 \qquad\text{for every }i.
\]

Suppose, toward a contradiction, that an allocation $(A_i)_{i=1}^n$
satisfies $c_i(A_i)\le\mu+\ell$ for every $i$. We first show that
$w(A_i)=T$ for every $i$.

Recall that there are exactly $n$ high items, namely the items
of types $Z,P,Q,D$. Each has auxiliary weight at least $3T/4$
and correction at most $1$. Hence a bundle containing two high
items costs at least
\[
 2\left(\frac{9T}{4}-1\right)>\mu+\ell.
\]
Thus each allocated bundle contains at most one high item,
and consequently exactly one. Its correction is therefore at
most $k+1$: at most $1$ from its high item and at most $k$ from
selected $Y$-items.

All auxiliary weights are multiples of $k$. If $w(A_i)>T$, then
\[
 c_i(A_i)
 \ge 3(T+k)-(k+1)
 =\mu+2k-1
 >\mu+\ell,
\]
a contradiction. Thus $w(A_i)\le T$ for every $i$.
Since $w(M)=nT$, equality holds throughout.

It follows that the same allocation, viewed as a goods allocation,
satisfies
\[
 v_i(A_i)
 =6w(A_i)-c_i(A_i)
 =2\mu-c_i(A_i)
 \ge\mu-\ell
 \qquad\text{for every }i.
\]
This contradicts the goods-case hardness proposition (Proposition~\ref{prop:robust}).

Therefore, every chores allocation gives some agent cost
greater than $\mu+\ell$. The parameter estimates from the goods
proof give
\[
 \frac{\ell}{\mu}
 =\frac{1}{4896k^2}
 >\frac{1}{400000(\log n)^2},
\]
which concludes the proof.
\end{proof}

\section{Constant Hardness for Chores}
\label{sec:constant-chores}

We construct an instance with four agents and eleven chores in which
all agents have MMS equal to their proportional share, namely $90$,
but every allocation gives some agent cost at least $93$.
In particular, this establishes
$\beta^{\mathrm{MMS}}_4 \ge 93/90 = 31/30$.

\paragraph{Description of the instance.}
There are two agents with cost function $c_1$ and two agents with
cost function $c_2$. The chore set is
\[
M=\{a,b,c,d,e,f,g,h,i,j,k\}.
\]
The costs are additive and specified by the following table.
\[
\begin{array}{c|rrrrrrrrrrr}
\text{Item} &a&b&c&d&e&f&g&h&i&j&k\\ \hline
c_1 &19&40&31&37&16&37&25&40&25&31&59\\
c_2 &21&45&27&39&15&39&24&45&24&27&54
\end{array}
\]
Notice that
\[
c_1(M)=c_2(M)=360.
\]

\begin{proof}[Proof of Theorem~\ref{thm:constant-chores}]
For cost function $c_1$, the partition
\[
\bigl(\{a,b,c\},\{d,e,f\},\{g,h,i\},\{j,k\}\bigr)
\]
consists of four bundles of cost $90$. For cost function $c_2$,
the partition
\[
\bigl(\{d,g,j\},\{a,e,k\},\{b,h\},\{c,f,i\}\bigr)
\]
also consists of four bundles of cost $90$.
Together with the average-cost lower bound, these partitions give
\[
\mathrm{MMS}^{\mathrm{ch}}_i
=\mathrm{PROP}^{\mathrm{ch}}_i
=\frac{360}{4}=90
\qquad\text{for every agent }i.
\]

The exhaustive verification in Appendix~\ref{app:code}
computes
\[
\min_{(B_1,\ldots,B_4)\in\mathcal{A}_4(M)}
\max_{i\in[4]}c_i(B_i)=93.
\]
Consequently, every allocation gives some agent cost at least a
$93/90=31/30$ factor of her MMS, and hence
\[
\beta^{\mathrm{MMS}}
\ge \beta^{\mathrm{MMS}}_4
\ge \frac{31}{30},
\]
which concludes the proof.
\end{proof}

\section*{Acknowledgments}

This research was supported by the Israel Science Foundation (grant No. 3256/26).

\paragraph{Use of AI tools.}
The authors used ChatGPT (OpenAI) to explore and refine
impossibility constructions, assist with the development
and checking of proof arguments, and develop computational
search and verification code. This assistance contributed
to the development and presentation of the results in
Sections~2--5. ChatGPT was also used to search for related
literature and to draft and revise portions of the manuscript.
The authors take full responsibility for all content,
including the mathematical claims, proofs, computational
results, and references.


\printbibliography

\appendix
\section{Missing Proofs}
\begin{proof}[Proof of Lemma~\ref{lem:code}]
We show the existence of such a code by greedily adding balanced codewords and removing the words at distance
less than $k/4$.

At each step, for the added codeword  $u$, the number of  words in $L_k$ of distance $2t$ from $u$ is $\binom{k/2}{t}^{\!2}$ (the distance must be even). Thus, the number of removed words at each step is at most 
\begin{equation}
\label{eq:bound1}    
 \sum_{t=0}^{k/8}\binom{k/2}{t}^{\!2}
 \le\left(\sum_{t=0}^{k/8}\binom{k/2}{t}\right)^2,
\end{equation}

Let $H_2(p)=-p\log p-(1-p)\log(1-p)$ be the standard entropy function. 
It holds that \begin{equation}
\label{eq:bound2}
 \sum_{t=0 }^{ k/8}\binom{k/2}t
 \le 3^{k/8}\sum_{t=0}^{k/2}\binom{k/2}t3^{-t}
 =3^{k/8}(4/3)^{k/2}=2^{kH_2(1/4)/2}.
\end{equation}
Moreover,
\begin{equation}
\label{eq:bound3}    
 H_2(1/4)=\frac12+\frac34\log(4/3)<\frac{13}{16},
\end{equation}
since $(4/3)^{12}<2^5$.

The size of $L_k$ is at least $2^k/(k+1)$. Combining Equations~\eqref{eq:bound1}, \eqref{eq:bound2}, and \eqref{eq:bound3} implies that the constructed greedy code has size at least
\[
 \frac{2^{3k/16}}{k+1}\ge2^{k/8},
\]
where the last inequality follows from $k+1\le2^{k/16}$ for $k\ge128$.
\end{proof}

\section{Computational Verification of the Constant-Gap Instances}
\label{app:code}
\subsection{Code Verification for Goods}

The following Python code enumerates
all $4^{11}$ allocations of the eleven items
among the four agents. For each allocation, it computes the
minimum value received by an agent and records the largest
such minimum over all allocations. 

\begin{lstlisting}[
    language=Python,
    basicstyle=\ttfamily\footnotesize,
    keywordstyle=\bfseries,
    commentstyle=\itshape,
    numbers=left,
    numberstyle=\tiny,
    numbersep=8pt,
    breaklines=true,
    showstringspaces=false,
    keepspaces=true,
    columns=fullflexible,
    tabsize=4,
    captionpos=b,
    caption={Exhaustive verification of the constant-gap goods instance.},
    label={lst:constant-gap-verification}
]
from itertools import product

# Items are ordered as a, b, c, d, e, f, g, h, i, j, k.
values_1 = (1, 17, 3, 10, 1, 10, 7, 7, 7, 3, 18)
values_2 = (1, 14, 5, 10, 1, 9, 6, 7, 7, 5, 19)

# There are two agents of each type.
values = [values_1, values_1, values_2, values_2]

n_goods = 11
n_agents = 4
best_min = 0

for assignment in product(range(n_agents), repeat=n_goods):
    totals = [0] * n_agents
    for good, agent in enumerate(assignment):
        totals[agent] += values[agent][good]
    min_value = min(totals)
    if min_value > best_min:
        best_min = min_value

print("Best minimum value:", best_min)
\end{lstlisting}

The program outputs:
\begin{verbatim}
Best minimum value: 20
\end{verbatim}

\subsection{Code Verification for Chores}

The following Python code enumerates all $4^{11}$ allocations of the
eleven chores among the four agents. For each allocation, it computes
the maximum cost incurred by an agent and records the smallest such
maximum over all allocations. 

\begin{lstlisting}[
    language=Python,
    basicstyle=\ttfamily\footnotesize,
    keywordstyle=\bfseries,
    commentstyle=\itshape,
    numbers=left,
    numberstyle=\tiny,
    numbersep=8pt,
    breaklines=true,
    showstringspaces=false,
    keepspaces=true,
    columns=fullflexible,
    tabsize=4,
    captionpos=b,
    caption={Exhaustive verification of the constant-gap chores instance.},
    label={lst:constant-gap-verification-chores}
]
from itertools import product

# Items are ordered as a, b, c, d, e, f, g, h, i, j, k.
values_1 = (19, 40, 31, 37, 16, 37, 25, 40, 25, 31, 59)
values_2 = (21, 45, 27, 39, 15, 39, 24, 45, 24, 27, 54)

# There are two agents of each type.
values = [values_1, values_1, values_2, values_2]

n_goods = 11
n_agents = 4
best_max = 360

for assignment in product(range(n_agents), repeat=n_goods):
    totals = [0] * n_agents
    for good, agent in enumerate(assignment):
        totals[agent] += values[agent][good]
    max_value = max(totals)
    if max_value < best_max:
        best_max = max_value
 
print("Best maximum value:", best_max)
\end{lstlisting}

The program outputs:
\begin{verbatim}
Best maximum value: 93
\end{verbatim}

\end{document}